\documentclass[11pt,letterpaper]{article}

\usepackage[utf8]{inputenc}
\usepackage[english]{babel}

\usepackage[dvipsnames,usenames]{xcolor}
\usepackage[colorlinks=true,pdfpagemode=UseNone,urlcolor=RoyalBlue,linkcolor=RoyalBlue,citecolor=OliveGreen,pdfstartview=FitH]{hyperref}

\usepackage{subcaption}

\usepackage[margin=1in]{geometry}
\usepackage{amsfonts,amsmath,amsthm,thmtools}
\usepackage{bm}
\usepackage{nicefrac}
\usepackage[ruled, vlined, linesnumbered]{algorithm2e}

\declaretheorem{theorem}
\declaretheorem[sibling=theorem]{lemma}

\declaretheorem[style=definition]{definition}
\declaretheorem[style=definition]{example}

\usepackage{color}
\usepackage{todonotes}
\usepackage{soul}

\usepackage{booktabs}
\usepackage{caption}

\usepackage[numbers, sort]{natbib}

\usepackage{tcolorbox}

\usepackage{enumitem}

\renewcommand{\vec}[1]{\bm{#1}}

\newcommand{\dif}[1]{\mathop{\mathrm{d}#1}\nolimits}

\title{
    Online Nash Welfare Maximization Without Predictions
    \thanks{This is the second version of the paper on arXiv, which improves exposition based on feedbacks that we received after posting the first version. The main results stay the same. We thank anonymous reviewers for insightful and constructive comments, and thank Arindam Khan for pointing out several inaccurate references to previous work.}
}

\author{
    Zhiyi Huang
    \thanks{University of Hong Kong. Email: zhiyi@cs.hku.hk, xkshu@cs.hku.hk.}
    \and
    Minming Li
    \thanks{City University of Hong Kong. Email: minming.li@cityu.edu.hk, t.z.wei-8@my.cityu.edu.hk.}
    \and
    Xinkai Shu
    \footnotemark[2]
    \and
    Tianze Wei
    \footnotemark[3]
}
\date{February 2023}

\begin{document}

\begin{titlepage}
    \thispagestyle{empty}
    \maketitle
    \begin{abstract}
        \thispagestyle{empty}
        The maximization of Nash welfare, which equals the geometric mean of agents' utilities, is widely studied because it balances efficiency and fairness in resource allocation problems.
Banerjee, Gkatzelis, Gorokh, and Jin (2022) recently introduced the model of online Nash welfare maximization for $T$ divisible items and $N$ agents with additive utilities with predictions of each agent's utility for receiving all items.
They gave online algorithms whose competitive ratios are logarithmic.
We initiate the study of online Nash welfare maximization \emph{without predictions}, assuming either that the agents' utilities for receiving all items differ by a bounded ratio, or that their utilities for the Nash welfare maximizing allocation differ by a bounded ratio.
We design online algorithms whose competitive ratios only depend on the logarithms of the aforementioned ratios of agents' utilities and the number of agents.

    \end{abstract}
\end{titlepage}

\section{Introduction}
\label{sec:introduction}

Suppose that you run a daycare center.
From time to time, some organizations donate candies for children who have different preferences over various kinds of candies:
Alice likes chocolates, Bob prefers gummy bears, Charlie favors lollipops, and so forth.
The candies come in abundance, so you may consider it as a resource allocation problem of \emph{divisible} items.
How would you distribute the candies to the children?

Naturally, we would like to allocate the candies based on the children's values for them.
We will assume that the children have \emph{additive utilities} for receiving bundles of candies.
Following this idea, it may be tempting to allocate the candies to the children in a way that maximizes the sum of children's utilities for the candies that they received, a.k.a., the social welfare.
The example below, however, highlights the potential unfairness of this approach.

\begin{example}
    Consider distributing two packs of chocolates and two packs of gummy bears to two children Alice and Bob.
    Alice has values $100$ and $15$ for receiving a pack of chocolates and gummy bears respectively.
    Bob, on the other hand, has values $1$ and $10$ for chocolates and gummy bears respectively.
\end{example}

The social welfare maximizing allocation gives everything to Alice for her higher values, and leaves nothing for Bob.
This is blatantly unfair, especially at a daycare center.

What if we distribute evenly and let each child have a pack of chocolates and a pack of gummy bears?
While this is fair, it is inefficient.
The gummy bears contribute little to Alice's utility, while the chocolates do not add much to Bob's utility.

The question then becomes how to allocate the candies if we want to strike a balance between efficiency and fairness.
In the above example, allocating both packs of chocolates to Alice and both packs of gummy bears to Bob is a decent option.
In general, maximizing the \emph{Nash welfare}, i.e., the geometric mean of the children's utilities, is a good proxy for balancing efficiency and fairness.
Indeed, the allocation that maximizes Nash welfare also satisfies several standard notions of fairness, such as envy-freeness~\cite{CaragiannisKMPSW:2019} and proportionality~\cite{Vazirani:AGT:2007}.
It also coincides with the aforementioned decent allocation in the above example.
Further, the Nash welfare maximization problem is \emph{scale invariant}:
if we increase an agent's values for all items by the same multiplicative factor, the Nash welfare maximizing allocation would remain the same.
Moreover, maximizing the Nash welfare with divisible items and additive utilities reduces to solving the Eisenberg-Gale convex program~\cite{EisenbergG:AnnaMathStat:1959}.
We can therefore compute this allocation in polynomial time.

There is one more complication though.
We cannot predict when the donations of candies will be made, what kinds of candies will be donated, and in what quantities.
When some organizations make a donation, we need to distribute the candies to the children immediately without much information about future donations.

Besides the toy example above, many real online resource allocation problems for heterogeneous agents require balance between fairness and efficiency, e.g., allocating food each day to people in need~\cite{prendergast2017food}, or allocating shared computational resources to users~\cite{hao2016online}. 
Hence, our task becomes designing online allocation algorithms for Nash welfare maximization,
whose performance is evaluated by \emph{competitive ratio}, that is, the worst ratio between the online algorithm's solution and the offline optimal solution.
From now on, we will more generally talk about items and agents.

This problem was first studied by \citet*{BanerjeeGGJ:SODA:2022}.
They noticed that under worst-case analysis, no algorithm can have a competitive ratio better than the trivial $O(N)$, where $N$ is the number of agents.
They then turned to the model of online algorithms with predictions, in which the online algorithm is further given some predictions of each agent's utility for receiving the set of all items, which they called the \emph{monopolist utility}.
If the predictions were accurate, their algorithm would be $O(\log N)$-competitive.%
\footnote{\citet{BanerjeeGGJ:SODA:2022} also gave an $O(\log T)$-competitive ratio where $T$ is the number of items, under the assumption that items have unit supplies. Since this paper considers arbitrary supplies, the dependence in $T$ is no longer valid.}
They further proved that the competitive ratio is asymptotically optimal for the problem.

\subsection{Our Contribution}

Coming up with accurate predictions is not an easy task at all.
Is it truly hopeless to design good online algorithms \emph{without predictions}?
The hard instance which \citet{BanerjeeGGJ:SODA:2022} designed to rule out the existence of good online algorithms without predictions requires the agents' values to differ by an exponential factor, a rare phenomenon in real resource allocation problems.
Can we design online algorithms for ``natural instances'' where agents' values are not as extreme?

We introduce the notions of \emph{balanced} and \emph{impartial} instances.
An instance is $\lambda$-\emph{balanced} if the agents' monopolist utilities differ by at most a multiplicative factor $\lambda$.
This is closely related to the model of online algorithms with predictions by \citet{BanerjeeGGJ:SODA:2022}.
If we had accurate predictions of the agents' monopolist utilities, we could then normalize the agents' values accordingly to obtain a $1$-balanced instance using the fact that Nash welfare maximization is scale invariant.
If we had $(\alpha, \beta)$-approximate predictions, i.e., if they were at worst an $\alpha$-factor larger or a $\beta$-factor smaller than the agents' monopolist utilities, the normalized instance would be $\alpha \beta$-balanced.
Therefore, any competitive online algorithm for balanced instances further implies a competitive online algorithm with predictions.

Our main result for $\lambda$-balanced instances is a $\log^{1+o(1)} (\lambda N)$-competitive online algorithm.
We further show that this competitive ratio is nearly optimal in the sense that no online algorithm can be better than $\log^{1-o(1)} (\lambda N)$ competitive.
By the aforementioned reduction, this also implies a $\log^{1+o(1)} (\alpha \beta N)$-competitive online algorithm for the model with $(\alpha, \beta)$-approximate predictions, which improves the $O(\alpha \log \beta N)$-competitive online algorithm given by \citet{BanerjeeGGJ:SODA:2022} when $\alpha$ is sufficiently large.

On the other hand, we say that an instance is \emph{$\mu$-impartial} if the agents' utilities for the Nash welfare maximizing allocation differ by at most a multiplicative factor $\mu$.
For $\mu$-impartial instances, our main result is a $\log^{2+o(1)} \mu$-competitive online algorithm.
We remark that this competitive ratio is independent of the number of agents or items, and thus this would be an $O(1)$-competitive algorithm if the instance is $O(1)$-impartial.
In the special case when the agents' values for different items are binary, we have an improved competitive ratio $\log^{1+o(1)} \mu$;
we further prove that it is nearly optimal.

See Table~\ref{tab:summary} for a summary of our results.

\begin{table}[t]
    \centering

    \caption{Summary of results. Here $\lambda^*$ and $\mu^*$ are the smallest numbers for which the instance is $\lambda^*$-balanced and $\mu^*$-impartial. We omit the constants and $\log\log$ factors for brevity. We also omit the trivial bound of $N$.}
    \label{tab:summary}
    
    \begin{tabular}{p{3cm}p{6.5cm}p{3cm}}
        \toprule
         & Algorithm & Lower Bound \\
        \midrule
        $\lambda^*$-Balanced & $\log \lambda^* N$ & $\log \lambda^* N$ \\
        $\mu^*$-Impartial & $\min \big\{ \log \mu^* N \,,~ \log^2 \mu^* \big\}$ & $\log \mu^*$ \\
        $(\alpha, \beta)$-Predictions & $\alpha  \log \beta N$ \cite{BanerjeeGGJ:SODA:2022} $\to~ \log \alpha \beta N$ &  \\
        \bottomrule
    \end{tabular}
    
\end{table}

\subsection{Other Related Works}

There is a vast literature on online algorithms for resource allocation problems, such as online packing (c.f., \citet{AlonAABN:TALG:2006} and \citet{BuchbinderN:MOR:2009}) and online matching problems (c.f., \citet*{KarpVV:STOC:1990} and \citet{Mehta:FTTCS:2013}).
Most related to this paper is the work by \citet{DevanurJ:STOC:2012} on online matching with concave returns.
Maximizing the Nash welfare can be equivalently formulated as maximizing the sum of logarithms of the agents' utilities, a special case of the model of \citet{DevanurJ:STOC:2012} when the return functions are the logarithm.
However, their results do not imply competitive online algorithms for our problem because they studied the multiplicative competitive ratio with respect to (w.r.t.) the sum of log-utilities.
By contrast, \citet{BanerjeeGGJ:SODA:2022} and this paper consider the multiplicative competitive ratio w.r.t.\ the geometric mean of utilities, which corresponds to the additive approximation factor w.r.t.\ the sum of log-utilities.
\citet{BarmanKM:AAAI:2022} studied online algorithms for the more general $p$-mean welfare maximization which captures the Nash welfare as a special case when $p = 0$.
Their result for Nash welfare can be interpreted as an $O(\log^3 N)$-competitive algorithm for $1$-balanced instances;
by contrast our algorithm is $O(\log^{1+o(1)} (\lambda N)$-competitive for general $\lambda$-balanced instances. 

Nash welfare was introduced about 70 years ago \cite{Nash:1950, KanekoM:1979}. Its maximization has been extensively studied in the offline setting.
The problem of divisible items and additive utilities is equivalent to the Eisenberg-Gale convex program~\cite{EisenbergG:AnnaMathStat:1959}.
Much effort has been devoted to the problem of indivisible items and various kinds of utilities.
Even for additive utilities, the problem is NP-hard~\cite{Moulin:2004} and APX-hard~\cite{Lee:IPL:2017}.
\citet{ColeG:STOC:2015} gave the first constant-approximation algorithm for additive utilities.
Subsequently, \citet{AnariOSS:ITCS:2017} applied the theory of stable polynomials to get an $e$-approximation algorithm.
\citet{ColeDGJMVY:EC:2017} further improves the analysis to show a $2$-approximation.
Finally, \citet*{BarmanKV:EC:2018} gave a new algorithm that achieves the state-of-the-art $e^{1/e} < 1.45$-approximation.
The problem with more general utilities has also been studied. 
\citet*{GargHM:SODA:2018} designed a constant-approximation algorithm for budget-additive utilities.
\citet{AnariMOV:SODA:2018} introduced a constant-approximation algorithm for separable, piecewise-linear concave utilities.
For submodular utilities, \citet*{GargKK:SODA:2020} proposed an $O(n \log n)$-approximation algorithm, and \citet{LiV:FOCS:2022} recently obtained the first constant-approximation algorithm for submodular utilities.
Last but not least, \citet{BarmanBKS:ESA:2020} and \citet*{ChaudhuryGM:AAAI:2021} independently developed $O(n)$-approximation algorithms for the even more general subadditive utilities.

\citet*{AzarBJ:ESA:2010} studied an online resource allocation problem whose competitive algorithms are also competitive w.r.t.\ Nash welfare. 
However, they assumed that any agent's values for different items could only differ by a bounded ratio;
their competitive ratio is logarithmic in the numbers of agents and items as well as this bounded ratio.
By contrast, our notions of balance and impartiality compare different agents' utilities.
Their assumption rules out the possibility that an agent may have zero value for some item. Therefore, their results do not apply to our model. 
Besides that, there is a significant line of work about online resource allocation, which focuses on different objectives, e.g., \citet{gkatzelis2021fair} considered maximizing utilitarian social welfare with envy-freeness in online resource allocation. \citet{banerjee2022proportionally} studied an online public items allocation problem with predictions, where the objective is to achieve proportional fairness.



\clearpage

\section{Preliminaries}
\label{sec:prelim}

\subsection{Nash Welfare Maximization}

Consider $N$ agents and $T$ divisible items.
Denote the supply of an item $t$ by $s_t$.
We represent an allocation of items to agents by a non-negative real vector $\vec{x} = \big(x_{it}\big)_{1 \le i \le N, 1 \le t \le T}$, where $x_{it}$ is the amount of item $t$ allocated to agent $i$.
A feasible allocation must comply with the supply constraints of items, that is, it must satisfy $\sum_{i = 1}^N x_{it} \le s_t$ for any item $t$.
For each agent $1 \le i \le N$ and each item $1 \le t \le T$, $v_{it} \ge 0$ denotes agent $i$'s value for receiving a unit of item $t$.
The agents' utilities are \textit{additive}:
If an agent $i$ receives $x_{it}$ amount of each item $t$, its utility equals $u_i = \sum_{t = 1}^{T} x_{it} v_{it}$.

\begin{definition}
    The \textit{Nash welfare} of an allocation $\vec{x} = (x_{it})_{1 \le i \le N, 1 \le t \le T}$ is the geometric mean of the agents' utilities:
    \begin{equation*}
        \left(\prod_{i = 1}^{N} u_i\right)^{1/N} ~ =  \quad \prod_{i = 1}^{N} \left( \sum_{t = 1}^{T} x_{it} v_{it} \right)^{1/N}.
    \end{equation*}    
\end{definition}

We also consider an equivalent objective, namely, maximizing the sum of logarithms of utilities:
\begin{equation*}
    \max ~ \sum_{i=1}^{N} \log u_{i}.
\end{equation*}

Let $\vec{x^*}=(x_{it}^*)_{1 \le i \le N, 1 \le t \le T}$ denote the allocation that maximizes the Nash welfare, breaking ties arbitrarily.
Let $\vec{u^*}=(u_{i}^*)_{1 \le i \le N}$ denote the corresponding utilities of the agents.

Maximizing the Nash welfare is a good proxy for balancing the efficiency and fairness of the allocation.
Concretely, when the items are divisible, it implies that any agent would get at least $\frac{1}{N}$ of its monopolist utility for receiving all items, a notion of fairness known as proportionality.

\begin{lemma}[c.f., \citet*{Vazirani:AGT:2007}]
    \label{lem:nsw-proportionality}
    The Nash welfare maximizing allocation $\vec{x^*}$ and the corresponding utilities $\vec{u^*}$ satisfy that for any agent $1 \le i \le N$:
    \[
        u_i^* \ge \frac{1}{N} \sum_{t=1}^{T} s_t v_{it}
        ~.
    \]
\end{lemma}

\begin{proof}
    Suppose for contrary that there is an agent $i$ such that $u_i^* < \frac{1}{N} \sum_{t=1}^{T} s_t v_{it}$.
We argue that it would be beneficial to allocate more to agent $i$.
Concretely, consider an alternative allocation $\vec{x'} = (x'_{it})_{1 \le i \le N, 1 \le t \le T}$ defined as:
\[
    x'_{jt} =
    \begin{cases}
        \displaystyle \frac{x_{jt}^*}{1 + \varepsilon} & \mbox{, if $j \ne i$;} \\[2.5ex]
        \displaystyle \frac{x_{it}^* + \varepsilon s_t}{1 + \varepsilon} & \mbox{, if $j = i$.}
    \end{cases}
\]

This is a feasible allocation because for any item $1 \le t \le T$:
\[
    \sum_{j=1}^N x'_{jt} = \frac{1}{1+\varepsilon} \left( \sum_{j \ne i}^N x_{jt} + \big( x_{it} + s_t \varepsilon \big) \right)
    \le \frac{1}{1+\varepsilon} \left( s_t + s_t \varepsilon \right) = s_t
    ~.
\]

Further comparing the agents' utilities to the counterparts in the assumed optimal allocation, we have:
\[
    u'_i = \sum_{t=1}^T x'_{it} v_{it} = \frac{1}{1+\varepsilon} \left( \sum_{t=1}^T x^*_{it} v_{it} + \varepsilon \sum_{t=1}^T s_t v_{it} \right) = \frac{1}{1+\varepsilon} \left( u^*_i + \varepsilon \sum_{t=1}^T s_t v_{it} \right)
    ~,
\]
and for any agent $j \ne i$ we have $u'_j = \frac{1}{1+\varepsilon} u^*_j$.
Hence, we get a contradiction that:
\[
    \frac{\left(\prod_{i = 1}^{N} u'_i\right)^{1/N}}{\left(\prod_{i = 1}^{N} u_i^*\right)^{1/N}} ~ = ~ \frac{\left( 1 + \varepsilon \frac{\sum_{t=1}^{T} s_t v_{it}}{u_i^*}\right)^{1/N}}{1 + \varepsilon} > 1
\]
for a sufficiently small $\varepsilon$ by the assumption of $u_i^* < \frac{1}{N} \sum_{t=1}^{T} s_t v_{it}$.
\end{proof}

\subsection{Online Algorithms}

This paper considers an online setting in which the items arrive one by one, while the agents are known from the beginning.
When each item arrives, the algorithm observes the item's supply and the agents' values for the item.
It must then decide how to allocate the item immediately and irrevocably.
We consider the ratio of the optimal Nash welfare to the expected Nash welfare of the algorithm's allocation.
The \emph{competitive ratio} of an algorithm is the maximum of this ratio over all possible instances.

\paragraph{Greedy Algorithms with Anticipated Utilities.}

A natural strategy to allocate an item $t$ is to first estimate how much utility each agent $i$ could get from the other items, denoted as $u'_{it}$ and referred to as agent $i$'s anticipated utility.
We may then allocate item $t$ greedily to maximize the Nash welfare based on the anticipated utilities.
For example, a natural yet conservative form of anticipated utilities in our model would be the agents' utilities from the previous items.
In other models, the anticipated utilities may depend on some prior knowledge of the agents' values; for instance, the algorithm of \citet{BanerjeeGGJ:SODA:2022} may be viewed as a greedy algorithm with anticipated utilities that depend on the agents' monopolist utilities.

It is therefore instructive to first examine some basic properties of this family of algorithms.
We may formulate the allocation problem conditioned on anticipated utilities $u'_{it}$ as a convex program:
\begin{align}
    \label{eqn:greedy-predicted-utility}
    \text{maximize} & \quad \sum_{i=1}^N \log(u'_{it} + v_{it} x_{it})
    \\
    \notag
    \text{subject to} & \quad \sum_{i=1}^N x_{it} \le s_t \quad \mbox{and} \quad x_{it} \ge 0 \mbox{ for any agent $i$.}
\end{align}

Since the objective is monotone in $x_{it}$'s, the optimal solution to the above program, denoted as $\vec{z}_t = (z_{it})_{1 \le i \le N}$, must satisfy the first constraint with equality, i.e.:
\begin{equation}
    \label{eqn:greedy-predicted-utility-allocate-all}
    \sum_{i=1}^N z_{it} = s_t
    ~.
\end{equation}

Further, for the optimal multiplier $\nu^* \ge 0$ of the first constraint and for any agent $1 \le i \le N$, the optimality conditions imply that:
%
%
\begin{equation}
    \label{eqn:greedy-predicted-utility-optimal}
    z_{it} = \max \left\{ \frac{1}{\nu^*}-\frac{u'_{it}}{v_{it}} ~, ~ 0 \right\}
    ~.
\end{equation}
We give the following lemma to bound the gain of the logarithm of Nash welfare when we follow the optimal solution from Eqn.~\eqref{eqn:greedy-predicted-utility-optimal} to allocate an item $t$.

\begin{lemma}
    \label{lem:greedy-predicted-utility-increase}
    The optimal solution $\vec{z}$ to convex program \eqref{eqn:greedy-predicted-utility} satisfies that:
    \[
        \sum_{i=1}^{N} \log(u'_{it} +v_{it} z_{it}) - \sum_{i=1}^{N} \log u'_{it} \geq s_t \cdot \max_{1 \le i \le N} \frac{v_{it}}{u'_{it} + v_{it} z_{it}}
        ~.
    \]
\end{lemma}

\begin{proof}
    By the concavity of logarithm, the left-hand side is at least:
    \[
        \sum_{i=1}^N \frac{v_{it} z_{it}}{u'_{it} + v_{it} z_{it}}
    \]

    Further by Eqn.~\eqref{eqn:greedy-predicted-utility-optimal}, if $z_{it} > 0$ then $i$ maximizes $\frac{v_{it}}{u'_{it} + v_{it} z_{it}}$ (with maximum value $\nu^*$).
    Hence, the above equals:
    \[
        \max_{1 \le i \le N} \frac{v_{it}}{u'_{it} + v_{it} z_{it}} ~ \sum_{i=1}^N z_{it}
        ~.
    \]

    The lemma now follows by Eqn.~\eqref{eqn:greedy-predicted-utility-allocate-all}.
\end{proof}

\paragraph{Randomized versus Deterministic Algorithms.}
For divisible items, the best competitive ratios achievable by randomized and deterministic online algorithms are the same.
Given any randomized algorithm, we may convert it into a deterministic one such that for each item $t$, the amount of item $t$ that the deterministic algorithm allocates to any agent $i$ equals the expected amount that the randomized algorithm would allocate.
Since the Nash welfare is a concave function of the allocation, the deterministic algorithm yields a weakly larger Nash welfare.
Therefore, our positive results will describe online algorithms in their randomized form if that simplifies the presentation.
In our negative results, on the other hand, it suffices to consider deterministic algorithms.

\subsection{Balanced and Impartial Instances}

\citet{BanerjeeGGJ:SODA:2022} showed that for arbitrary instances, no online algorithm can have a competitive ratio better than the trivial $O(N)$.
They then explored the model of online algorithms with predictions, where the algorithms know the agents' monopolist utilities for receiving all items.

We observe that the agents' values in the hard instances of \citet{BanerjeeGGJ:SODA:2022} are lopsided:
the largest value is exponentially larger than the smallest.
This is rarely the case in real resource allocation problems.
Recall the example of allocating food and computational resources. It is almost impossible to see agents' values towards the same item differ substantially.
Hence this paper aims to design algorithms with good performance on ``natural instances'' in which different agents are of ``similar importance''.
We define two notions of ``natural instances''.
Motivated by \citet{BanerjeeGGJ:SODA:2022}, our first notion compares the agents' monopolist utilities for receiving all items.

\begin{definition}
    \label{def:balanced-instance}
    An instance is $\lambda$-\emph{balanced} if for any agents $1 \le i, j \le N$:
    \[
    \sum_{t = 1}^{T} s_t v_{it} \le \lambda \sum_{t = 1}^{T} s_t v_{jt}
    ~.
    \]
    We shall refer to $\lambda^* = \frac{\max_{1 \le i \le N} \sum_{t = 1}^{T} s_t v_{it}}{\min_{1 \le i \le N} \sum_{t = 1}^{T} s_t v_{it}}$ as the \emph{balance ratio} of the instance.
\end{definition}

We remark again that competitive online algorithms for $\lambda$-balanced instances can be transformed into online algorithms with predictions in the model of \citet{BanerjeeGGJ:SODA:2022}.
Given prediction $\tilde{V}_i$ of each agent $i$'s monopolist utility $V_i = \sum_{t=1}^{T} s_t v_{it}$ that is $(\alpha, \beta)$-approximate, i.e., $\alpha V_i \ge \tilde{V}_i \ge \frac{1}{\beta} V_i$, we can normalize agent $i$'s values to be $v_{it}' = \frac{v_{it}}{\tilde{V}_i}$ for all items $t$ to get an $\alpha \beta$-balanced instance.

Our second notion examines whether different agents get similar utilities in the Nash welfare maximizing allocation.

\begin{definition}
    \label{def:impartial-instance}
    
    Given any instance, an allocation $\vec{x}$ is $\mu$-\emph{impartial} if the corresponding utilities $\vec{u}$ satisfy that for any agents $1 \le i, j \le N$:
    \[
    u_i \le \mu \cdot u_j
    ~.
    \]
    An instance is $\mu$-impartial if every Nash welfare maximizing allocation $\vec{x}^{*}$ is $\mu$-impartial.
    Let $\mu^*$ be the smallest value of $\mu$ such that the instance is $\mu$-impartial;
    we shall refer to $\mu^*$ as the \emph{impartiality ratio} of the instance.
    
\end{definition}

Impartiality has a simple yet important implication:
For each item, we may ignore the agents whose values for the item are much smaller than the highest value.
We make this precise with the next lemma.

\begin{lemma}
    \label{lem:maximum-mu-value}
    For any agent $i$ and any item $t$, if $x_{it}^*>0$ then $v_{it}\geq \frac{1}{\mu^{*}} \max_{1 \le j \le N}v_{jt}$.
\end{lemma}
\begin{proof}
    Suppose for contradiction that there are agent $i$ and item $t$ for which we have $x_{it}^*>0$ but $v_{it}<\frac{1}{\mu^{*}} \max_{1 \le j \le N}v_{jt}$.
We argue that we should have allocated less item $t$ to agent $i$.
Let $j$ be an agent with highest value $v_{jt}$ for item $t$.
Consider reallocating an $\varepsilon$ amount of item $t$ from agent $i$ to agent $j$.
We claim that the Nash welfare increases for a sufficiently small $\varepsilon$.
Since the other agents' utilities stay the same, it suffices to prove that:
\begin{equation*}
    (u_i^* - \varepsilon v_{it}) (u_j^* + \varepsilon v_{jt}) - u_i^* u_j^* = \varepsilon (u_i^* v_{jt} - u_j^* v_{it}) - \varepsilon ^2 v_{it} v_{jt}
\end{equation*}
is positive.
Since $v_{it} < \frac{1}{\mu^*} v_{jt}$ by our assumption for contradiction and $u_j^* \le \mu^* u_i^*$ by the definition of $\mu^*$, we get that $u_i^* v_{jt} - u_j^* v_{it} > 0$.
Hence, the above is positive for a sufficiently small $\varepsilon$.
\end{proof}

Finally, the balance ratio and impartiality ratio are within a factor $N$ from each other.

\begin{lemma}
    \label{lem:balance-impartial}
    For any instance, we have $\lambda^* \leq \mu^* N$ and $\mu^* \leq \lambda^* N$.
\end{lemma}

\begin{proof}
    We first prove that $\lambda^* \le \mu^* N$.
    Let $i$ and $j$ be the agents with the maximum and minimum values for receiving all items respectively.
    We have:
    \begin{align*}
        \sum_{t=1}^{T} s_t v_{it}
        &
        = \lambda^* \sum_{t=1}^{T} s_t v_{jt} \geq \lambda^* u_j^* \\
        &
        \ge \frac{\lambda^*}{\mu^*} u_i^*
        && \mbox{(definition of $\mu^*$)} \\
        &
        \ge \frac{\lambda^*}{\mu^* N} \sum_{t=1}^{T} s_t v_{it}
        ~.
        && \mbox{(Lemma~\ref{lem:nsw-proportionality})}
    \end{align*}
    
    Cancelling $\sum_{t=1}^{T} s_t v_{it}$ on both sides proves the inequality.

    Next, we show that $\mu^* \le \lambda^* N$ using a similar argument.
    Let $i$ and $j$ be the agents with the minimum and maximum utilities in the Nash welfare maximizing allocation respectively.
    We have:
    \begin{align*}
        \sum_{t=1}^{T} s_t v_{it}
        &
        \geq \frac{1}{\lambda^*} \sum_{t=1}^{T} s_t v_{jt}
        && \mbox{(definition of $\lambda^*$)} \\
        &
        \geq \frac{1}{\lambda^*} u_j^* = \frac{\mu^*}{\lambda^*} u_i^* \\
        &
        \ge \frac{\mu^*}{\lambda^* N} \sum_{t=1}^{T} s_t v_{it}
        ~.
        && \mbox{(Lemma~\ref{lem:nsw-proportionality})}
    \end{align*}
    
    Cancelling $\sum_{t=1}^{T} s_t v_{it}$ on both sides proves the inequality.
\end{proof}

In conclusion, these two notions reflect the ``naturality'' of an instance, as a replacement of prediction in previous works.
In Section \ref{sec:balance}, we study balanced instances and achieve $O(\log \lambda^* N)$-competitive ratio.
For impartial instances, besides the $O(\log \mu^* N)$ result as a corollary of Lemma \ref{lem:balance-impartial}, in Section \ref{sec:impartial} we give an online algorithm whose competitive ratio is $O(\log^2 \mu^*)$, only depending on the impartiality ratio. (All $\log\log$ factors are omitted here.)
Therefore, hopefully the impartiality ratio is the better one to use.

\section{Balanced Instances}
\label{sec:balance}

This section studies balanced instances.
Section~\ref{sec:half-greedy-bounded} introduces our algorithm under an additional assumption that we were given an upper bound of the balance ratio (Section~\ref{sec:half-greedy-bounded}).
Then, Section~\ref{sec:half-greedy-general} explains how to remove this assumption by guessing the balance ratio, while losing at most $\log\log$ factors in the competitive ratio.

\subsection{Algorithm with a Known Upper Bound of the Balance Ratio}
\label{sec:half-greedy-bounded}

Suppose that we are given an upper bound $\lambda$ of the instance's balance ratio $\lambda^*$.
In other words, we know that the instance is $\lambda$-balanced.
The algorithm will divide each item $t$'s supply equally into two halves.
On the one hand, it allocates the first half equally to all agents, a naïve strategy that is certainly fair but is not efficient enough to approximately maximize the Nash welfare on its own.
On the other hand, it greedily allocates the second half of the item to maximize the Nash welfare assuming that each agent would get not only their utilities for the second halves of the previous items allocated to them, but also a fraction of the sum of \emph{all agents}' monopolist utilities for all known items.
Concretely, let $\vec{z}_t = (z_{it})_{1 \le i \le N}$ denote the allocation of the second half of item $t$.
For anticipated utilities:
\begin{equation}
    \label{eqn:half-and-half-anticipated-utility}
    u'_{it} ~ = ~ \frac{1}{2 \lambda N^2} \underbrace{\sum_{j=1}^N \sum_{t'=1}^t v_{jt'} s_{t'}}_{\substack{\text{sum of monopolist utilities}\\ \text{for all known items,}\\ \textbf{including item $t$}}} + \underbrace{\vphantom{\sum_{j=1}^N} \sum_{t'=1}^{t-1} v_{it'} z_{it'}}_{\substack{\text{$i$'s utility for the second halves}\\ \text{of previous items allocated to $i$,}\\ \textbf{excluding item $t$}}}
\end{equation}
the algorithm chooses $\vec{z}$ to maximize $\sum_{i=1}^N \log \big( u'_{it} + v_{it} x_{it} \big)$ subject to $\sum_{i=1}^N z_{it} \le \frac{s_t}{2}$ and $z_{it} \ge 0$ for all agents $i$.
We call this algorithm Half-and-Half.
See Algorithm~\ref{alg:half-greedy-lambda} for a formal definition.

\begin{algorithm}[t]
    \caption{\textbf{Half-and-Half} (for  $\lambda$-balanced instances)}
    \label{alg:half-greedy-lambda}
    
    \For{\text{\rm each item $1 \le t \le T$}}
    {
        Let $y_{it} = \frac{s_t}{2N}$ for all agents $1 \le i \le N$.\\
        Let $z_{it}$ maximize (for anticipated utilities $u'_{it}$ defined in Eqn.~\eqref{eqn:half-and-half-anticipated-utility})
        \[
            \sum_{i=1}^N \log (u'_{it} + v_{it} z_{it})
        \]
        subject to $\sum_{i=1}^N z_{it} \le \frac{s_t}{2}$ and $z_{it} \ge 0$ for all agents $1 \le i \le N$.
        
        Allocate $x_{it} = y_{it}+z_{it}$ amount of item $t$ to each agent $1 \le i \le N$.
    }
\end{algorithm}

\begin{theorem}
    \label{thm:half-greedy-lambda}
    Algorithm \ref{alg:half-greedy-lambda} is $O(\log \lambda N)$-competitive.
\end{theorem}

Before getting into the proof of Theorem~\ref{thm:half-greedy-lambda}, a comparison with the Set-Aside Greedy algorithm of \citet{BanerjeeGGJ:SODA:2022} is warranted, since readers familiar with the previous algorithm may have noticed the similarity between the two algorithms.
Both algorithms divide each item equally into two halves; both allocate the first half equally, and the second half by some greedy algorithm with anticipated utilities.
The difference lies in the designs of anticipated utilities.
The Set-Aside Greedy algorithm may be viewed as replacing the first part of our anticipated utility in Eqn.~\eqref{eqn:half-and-half-anticipated-utility} by a $\frac{1}{2N}$ fraction of the prediction on agent $i$'s monopolist utility.
Following the idea of Set-Aside Greedy, a natural attempt is to use each agent $i$'s monopolist utility for all known items as the prediction of its final monopolist utility.
In the online setting, however, the items that contribute the most to an agent's monopolist utility may come at the very end.
In that case, the algorithm would underestimate the agent's final utility for a long time, and as a result might unnecessarily allocate many items to this agent in the early rounds.
Our solution is to aggregate the monopolist utilities of \emph{all agents} for the known items into an anticipated utility for \emph{every agent}, an idea driven by the assumption of balanced instances.

We next present the analysis of Algorithm~\ref{alg:half-greedy-lambda}.
It is useful to define the following auxiliary utilities for any agent $1 \le i \le N$ and any item $1 \le t \le T$:
\[
    \hat{u}_{it} ~ = ~ \frac{1}{2 \lambda N^2} \underbrace{\sum_{j=1}^N \sum_{t'=1}^t v_{jt'} s_{t'}}_{\substack{\text{sum of monopolist utilities}\\ \text{for items $1$ to $t$}}} + \underbrace{\vphantom{\sum_{j=1}^N} \sum_{t'=1}^t v_{it'} z_{it'}}_{\substack{\text{$i$'s utility for the second halves}\\ \text{of items $1$ to $t$ allocated to $i$}}}
    ~.
\]

By the underlying logic of the algorithm's greedy allocation of the second halves of the items, this is what the algorithm anticipates agent $i$'s utility to be after allocating item $t$.
The next lemma validates this anticipation at the end of the algorithm.

\begin{lemma}
    \label{lem:half-greedy-lambda-estimation}
    For each agent $i$, $u_i \geq \hat{u}_{iT}$.
\end{lemma}
\begin{proof}
    It suffices to show that agent $i$'s utility for the first halves of the items allocated to it is greater than or equal to the first part of $\hat{u}_{iT}$.
    By definition, agent $i$'s utility for the first halves is:
    \[
        \frac{1}{2N} \sum_{t=1}^T v_{it} s_t
        ~.
    \]

    Further by the assumption of balanced instances, the monopolist utility of any other agent is at most $\lambda$ times larger than agent $i$'s monopolist utility.
    Therefore, we have:
    \[
        \sum_{j=1}^{N}\sum_{t=1}^{T}v_{jt}s_t\leq \lambda N\sum_{t=1}^{T}v_{it}s_t
        ~.
    \]

    Combining the two claims proves the lemma.
\end{proof}

\begin{proof}[Proof of Theorem~\ref{thm:half-greedy-lambda}]
    Since the second halves of the items are allocated by a greedy algorithm with anticipated utilities, we may use Lemma~\ref{lem:greedy-predicted-utility-increase} in the analysis.
    For Lemma~\ref{lem:greedy-predicted-utility-increase} to be effective, however, we need the anticipated utilities to be good approximations of the agents' utilities in the Nash welfare maximizing allocation.
    Fortunately, we only need a polynomial approximation, which is satisfied sufficiently early so that allocating the remaining items correctly still yields approximately optimal Nash welfare.
    
    Let $t^*$ be the earliest item for which $\sum_{i = 1}^{N}\sum_{t = 1}^{t^*} v_{it} \geq \frac{1}{2} \min_{1 \le i \le N} u_i^*$.
    The choice of $t^*$ ensures two properties.
    First, the contribution of items from $t^*$ to $T$ to any agent $i$'s utility in the optimal solution is at least $\frac{1}{2} u_i^*$.
    In other words, even if we had made completely wrong allocations, resulting in zero utilities for all agents, we would have lost at worst half of the Nash welfare.
    Second, the anticipated utility is from now on at least a polynomial approximation of the minimum utility of an agent in the Nash welfare maximizing allocation.
    Further, the instance is at worst $N \lambda$-impartial according to Lemma~\ref{lem:balance-impartial}.
    Hence, this is in fact a polynomial approximation for every agent.

    When an item $t \ge t^*$ arrives, by Lemma \ref{lem:greedy-predicted-utility-increase} the allocation $\vec{z}_t = (z_{it})_{1 \le i \le N}$ of the second half of this item satisfies:
    %
    \begin{equation}
        \label{eqn:half-and-half-analysis-start}
        \sum_{i=1}^{N} \Big( \log \big( u'_{it} + v_{it} z_{it} \big) - \log u'_{it} \Big) \geq \frac{s_t}{2} \max_{1 \le i \le N} \frac{v_{it}}{u'_{it} + v_{it} z_{it}}
        ~.
    \end{equation}

    To relate our inequality to the Nash welfare maximizing allocation $\vec{x^*}$, we apply $\sum_{i=1}^N x^*_{it} \le s_t$ to lower bound the right-hand side of Eqn.~\eqref{eqn:half-and-half-analysis-start} by:
    \[
        \frac{1}{2} \sum_{i=1}^N \frac{v_{it} x^*_{it}}{u'_{it} + v_{it} z_{it}}
        ~.
    \]

    Further, by Lemma~\ref{lem:half-greedy-lambda-estimation} we have $u_i \ge \hat{u}_{iT}$ for any agent $i$, while $\hat{u}_{iT}$ is greater than or equal to $\hat{u}_{it}$ by the definition of $\hat{u}_{it}$'s.
    Putting together, we conclude that:
    \[
        \sum_{i=1}^{N} \Big( \log \big( u'_{it} + v_{it} z_{it} \big) - \log u'_{it} \Big)
        \ge 
        \frac{1}{2} \sum_{i=1}^N \frac{v_{it} x^*_{it}}{u_i}
        ~.
    \]

    Since $\hat{u}_{it} = u'_{it} + v_{it} z_{it}$ and $\hat{u}_{i(t-1)} \le u'_{it}$, we get telescopic cancellations summing over items from $t^*$ to $T$.
    Note that, however, we shall not apply this relaxation for $u'_{it^*}$ for a technical reason that shall be clear shortly.
    On the other hand, the numerator on the right-hand side sums to at least $\frac{u^*_i}{2}$ for all agents $i$ because of the choice of $t^*$.
    Hence:
    \[
        \sum_{i=1}^{N} \Big( \log \hat{u}_{iT} - \log u'_{it^*} \Big)
        \ge
        \frac{1}{4} \sum_{i=1}^N \frac{u^*_i}{u_i}
        ~.
    \]

    Further, by Lemma~\ref{lem:half-greedy-lambda-estimation} we have $u_i \ge \hat{u}_{iT}$.
    We also have:
    \begin{align*}
        u'_{it^*}
        &
        \ge \frac{1}{2 \lambda N^2} \sum_{i = 1}^{N} \sum_{t = 1}^{t^*} v_{it}
        &&
        \mbox{(definition of $u'_{it}$)} \\[1ex]
        &
        \ge \frac{1}{4 \lambda N^2} \min_{1 \le i \le N} u_i^*
        &&
        \mbox{(choice of $t^*$)} \\
        &
        \ge \frac{1}{4 \lambda^2 N^3} \left( \prod_{i=1}^N u_i^* \right)^{\frac{1}{n}}
        ~.
        &&
        \mbox{($\lambda N$-impartiality by Lemma~\ref{lem:balance-impartial})}
    \end{align*}

    Note that this inequality would not be true in general if we had relaxed $u'_{it^*}$ to $\hat{u}_{i(t^*-1)}$.

    Putting together and by AM-GM inequality, we get that:
    \[
        \log \prod_{i=1}^N u_i - \log \prod_{i=1}^N u^*_i + N \log 4 \lambda^2 N^3
        \ge
        \frac{1}{4} \sum_{i=1}^N \frac{u^*_i}{u_i}
        \ge \frac{N}{4} \left( \prod_{i=1}^N \frac{u^*_i}{u_i} \right)^{\frac{1}{N}}
        ~.
    \]

    Let $\Gamma = (\prod_{i=1}^N \frac{u^*_i}{u_i})^{\frac{1}{N}}$ be the ratio of the optimal Nash welfare to the algorithm's Nash welfare.
    The above inequality is equivalent to:
    \[
        - \log \Gamma + \log 4\lambda^2 N^3 \ge \frac{1}{4} \Gamma
        ~.
    \]
    
    Since $\Gamma \ge 1$ and thus $\log \Gamma \ge 0$, the above inequality implies $\Gamma \le 4 \log 4 \lambda^2 N^3 = O(\log \lambda N)$.
\end{proof}

\subsection{Guessing the Balance Ratio}
\label{sec:half-greedy-general}

When we have no prior knowledge of the balance ratio, we can guess an upper bound of the balance ratio by sampling from an appropriate distribution.
Since the final ratio depends logarithmically on the upper bound, it suffices to make a good enough guess that is at most a polynomial of the true balance ratio $\lambda^*$.

Concretely, we shall consider a sequence of numbers starting from $2$, such that each sequel number is the square of the previous number.
We will sample each number $\lambda$ with a probability that is inverse polynomial in $\log\log \lambda$;
this ensures that the correct guess is made with a sufficiently large probability.
Finally, we apply the prior-dependent Half-and-Half algorithm (Algorithm~\ref{alg:half-greedy-lambda}) with the guessed upper bound $\lambda$.
See Algorithm~\ref{alg:half-greedy-general} for a formal definition.

\begin{algorithm}[t]
    \caption{\textbf{Half-and-Half} (for instances with unknown balance ratio)}
    \label{alg:half-greedy-general}
    
    Sample $\lambda$ such that it equals $2^{2^k}$ with probability $\frac{6}{\pi^2} \cdot \frac{1}{(k+1)^2}$ for any non-negative integer $k$.
    
    Run Algorithm~\ref{alg:half-greedy-lambda} with $\lambda$ to allocate the items to the agents.
\end{algorithm}

\begin{theorem}
    \label{thm:half-greedy-general}
    Algorithm \ref{alg:half-greedy-general} is $O \big( \log\lambda^* N \, (\log \log \lambda^*)^2 \big)$-competitive.
\end{theorem}

\begin{proof}
    Suppose that $k$ is the smallest positive integer such that the balance ratio $\lambda^*$ is at most $2^{2^k}$.
    Then, we have $2^{2^{k-1}} < \lambda^* \le 2^{2^k}$ which further implies that $\log \log \lambda^* > k - 1$.
    Hence, the algorithm correctly guesses $\lambda = 2^{2^k}$ with probability at least $\Omega(\frac{1}{k^2}) = \Omega( \frac{1}{(\log\log \lambda^*)^2})$.
    When that happens, the algorithm's Nash welfare is an $O(\log \lambda^* N)$ approximation to the optimal Nash welfare.
    Therefore, even if the algorithm got zero Nash welfare from the other guesses, we would still have the stated competitive ratio.
\end{proof}

 Considering the relation of balance and impartiality ratios (Lemma~\ref{lem:balance-impartial}), Theorem~\ref{thm:half-greedy-general} can directly imply Theorem~\ref{thm:half-greedy-general-mu}, which means the algorithm is also competitive for impartial instances.
The next section will develop algorithms tailored for impartial instances with competitive ratios independent of the number of agents.

\begin{theorem}
\label{thm:half-greedy-general-mu}
    Algorithm \ref{alg:half-greedy-general} is $O \big( \log \mu^* N \, (\log \log \mu^* N)^2 \big)$-competitive.
\end{theorem}

\section{Impartial Instances}
\label{sec:impartial}

This section considers impartial instances.
Section \ref{sec:greedy-binary} examines a special case of binary values: an agent $i$'s value for an item $t$ is either $v_{it} = 0$, or some value $v_{it} = v_t$ that is the same for all agents.
We will show that a simple greedy algorithm, which was referred to as Myopic Greedy by \citet{BanerjeeGGJ:SODA:2022}, is $O(\log \mu^*$)-competitive in this special case.
This is nearly tight as the next section will show an almost matching lower bound.
Section \ref{sec:greedy-reduction-bounded} then explains how to reduce the general case to the binary case, under an additional assumption that we are given an upper bound $\mu$ of the impartiality ratio $\mu^*$, i.e., if we know that the instance is $\mu$-impartial.
Finally, Section \ref{sec:greedy-reduction-general} applies the same technique as in the previous section to randomly guess an upper bound, removing the additional assumption while losing a factor that depends on $\log \log \mu^*$ in the competitive ratio.

\subsection{Myopic Greedy and Binary Values}
\label{sec:greedy-binary}

The Myopic Greedy algorithm simply allocates each item $t$ greedily to maximize the Nash welfare conditioned on the allocation before item $t$.
In other words, it is a greedy algorithm with anticipated utilities, for which the anticipated utility of an agent equals its utility for the items allocated to it so far.
See Algorithm~\ref{alg:greedy-binary} for a formal definition.

~

\begin{algorithm}[H] 
    \caption{\textbf{Myopic Greedy}}
    \label{alg:greedy-binary}
    \For{\text{\rm each item $1 \le t \le T$}}
    {
        Let the allocation $\vec{x}_t = (x_{it})_{1 \le i \le N}$ maximize:
        \[
            \sum_{i=1}^N \log \Big( \sum_{t'=1}^{t-1} v_{it'} x_{it'} + v_{it} x_{it} \Big)
        \]
        subject to $\sum_{i=1}^N x_{it} \le s_t$ and $x_{it} \ge 0$ for all agents $1 \le i \le N$.
    }
\end{algorithm}

~

The rest of this subsection will assume that the agents' values are binary, that is, for any item $1 \le t \le T$ and any agent $1 \le i \le N$, either $v_{it} = 0$ or $v_{it} = v_t$.

\begin{theorem}
    \label{thm:greedy-binary}
    Algorithm~\ref{alg:greedy-binary} is $O(\log\mu^*)$-competitive if the agents have binary values for the items.
\end{theorem}

In fact, we will prove a slightly stronger result so that in the next subsection we can reduce the general case to the case of binary values.
We formulate the stronger claim as the next lemma.

\begin{lemma}
    \label{lem:binary-greedy-ratio}
    For any feasible allocation $\vec{\tilde{x}} = (\tilde{x}_{it})_{1 \le i \le N, 1 \le t \le T}$ and the agents' corresponding utilities $\vec{\tilde{u}} = (\tilde{u}_i)_{1 \le i \le N}$,
    if allocation $\vec{\tilde{x}}$ is $\tilde{\mu}$-impartial, then the Nash welfare of the allocation by Algorithm~\ref{alg:greedy-binary} is at least an $O(\log \tilde{\mu})$ approximation to the Nash welfare of allocation $\vec{\tilde{x}}$, i.e.:
    \[
        \left( \frac{\prod_{i=1}^N \tilde{u}_i}{\prod_{i=1}^N u_i} \right)^{\frac{1}{N}} = ~ O(\log \tilde{\mu})
        ~.
    \]
\end{lemma}

Theorem~\ref{thm:greedy-binary} follows as a corollary by letting $\vec{\tilde{x}}$ be the Nash welfare maximizing allocation $\vec{x^*}$.

The rest of the subsection focuses on proving Lemma~\ref{lem:binary-greedy-ratio}.
We assume without loss of generality that the agents are sorted by their utilities for the algorithm's allocation, i.e.:
\[
    u_1 \le u_2 \le \dots \le u_N
    ~.
\]

We start with the following lemma which links the agents' utilities $\vec{u}$ for the algorithm's allocation, and their utilities $\vec{\tilde{u}}$ for the benchmark allocation $\vec{\tilde{x}}$.
\begin{lemma}
    \label{lem:binary-inquality}
    For any $1 \le i \le N$:
    \begin{equation}
        \label{eqn:binary-inquality}
        \sum_{j=1}^{i}\tilde{u}_j\leq\sum_{j=1}^{N}\min\{u_i,u_j\}.
    \end{equation}
\end{lemma}

\begin{proof}
    For any fixed $1 \le i \le N$, let $S_i$ be the set of items for which at least one of agents $1$ to $i$ have positive values, i.e.:
    \[
        S_i = \Big\{ 1 \le t \le T : \exists 1 \le j \le i, v_{jt} = v_t \Big\}
        ~.
    \]

    Then, the left-hand side of Equation~\eqref{eqn:binary-inquality} is upper bounded by:
    \[
        \sum_{j=1}^i \tilde{u}_j \le \sum_{t \in S_i} s_t v_t
        ~.
    \]

    To prove the lemma, it remains to show that the items in $S_i$ contribute at most $\min \{ u_i, u_j \}$ to any agent $j$'s utility for the algorithm's allocation.
    For agents $1 \le j \le i$, this is trivially true because the stated utility bound is simply $u_j$.
    For any agent $j > i$, consider the last moment when it receives a positive amount of some item $t$ in $S_i$.
    It suffices to show that its utility is at most $u_i$ after the allocation of item $t$.
    Recall that the Myopic Greedy algorithm is a greedy algorithm with anticipated utilities.
    First by the assumption of binary values and by Equation~\eqref{eqn:greedy-predicted-utility-optimal}, we conclude that any agent who receives a positive amount of item $t$ must have the smallest utility among all agents with values $v_t$ for item $t$, after the allocation of item $t$.
    Meanwhile, one of the agents from $1$ to $i$ has value $v_t$ for the item since $t \in S_i$, and this agent has utility at most $u_i$.
    Therefore, agent $j$'s utility after the allocation of item $t$ is at most $u_i$ as desired.
\end{proof}

If we view the agents' utilities $\vec{u}$ for the algorithm's allocation as variables, Lemma~\ref{lem:binary-inquality} offers a set of linear inequalities that relates them with the benchmark utilities $\vec{\tilde{u}}$.
The next lemma shows that subject to these inequalities, the smallest Nash welfare w.r.t.\ $\vec{u}$ is achieved when the inequalities all hold with equalities.

\begin{lemma}
    \label{lem:binary-utility}
    For any non-negative $\vec{u} = (u_i)_{1 \le i \le N}$ and $\vec{\tilde{u}} = (\tilde{u}_i)_{1 \le i \le N}$ that satisfy the inequality in Equation~\eqref{eqn:binary-inquality}, we have:
    \begin{equation}
        \label{eqn:binary-utility}
        \sum_{i=1}^{N} \log u_i ~ \geq ~ \sum_{i = 1}^{N} \log\left( \sum_{j = 1}^{i}\frac{\tilde{u}_i}{N+1-j} \right)
        ~.
    \end{equation}
\end{lemma}

\begin{proof}
    For ease of exposition, we rewrite Equation~\eqref{eqn:binary-inquality} as:
    \begin{equation}
        \label{eqn:step-function}
        \left\{
        \begin{aligned}
            &u_1 \cdot N& \geq~ &\tilde{u}_1\\
            &u_1 \cdot N + (u_2 - u_1) \cdot (N - 1) &\geq~ &\tilde{u}_1+\tilde{u}_2\\
            &...\\
            &u_1 \cdot N + \sum_{i = 1}^{N-1} ((u_{i+1} - u_{i}) \cdot (N - i) ) &\geq &\sum_{i}^{N} \tilde{u}_i
        \end{aligned}
        \right.
    \end{equation}
    Along with the restrictions $u_1\leq u_2\leq \cdots \le u_N$, the feasible region of $(u_i)_{i\in[N]}$ is a polytope. Since $\sum_{i=1}^{N} \log u_i$ is concave, it reaches minimum at a vertex of the polytope, i.e., where exactly $n$ of the inequalities are equalities. Let $(\bar{u}_i)_{i\in[N]}$ be a solution with minimum $\sum_{i=1}^{N} \log u_i$. Suppose there exists $i$ such that $\bar{u}_i=\bar{u}_{i+1}$. Let $i^*$ be the smallest such $i$. The $i^*$-th inequality in Equation \eqref{eqn:step-function} is not tight.

    Define $(u'_1,u'_2,\cdots,u'_N)$ as
    \begin{equation}
        \left\{
        \begin{aligned}
           &u'_i = \bar{u}_i & \qquad &\forall i< i^*\\
           &u'_i = \bar{u}_i-\varepsilon & &i=i^*\\
           &u'_i = \bar{u}_i+\frac{\varepsilon}{N-i} & &\forall i>i^*
        \end{aligned}
        \right.
    \end{equation}
    for sufficiently small $\varepsilon>0$. $(u'_1,u'_2,\cdots,u'_N)$ satisfies Equation \eqref{eqn:step-function}, and $u'_1u'_2\cdots u'_N<\bar{u}_1\bar{u}_2\cdots\bar{u}_N$ holds since $\bar{u}_i$ is the smallest among $\bar{u}_i,\bar{u}_{i+1},\cdots,\bar{u}_N$, contradicting the minimality of $(\bar{u}_i)_{i\in[N]}$.

    Therefore, $u_i=\sum_{i' = 1}^{i}\frac{\tilde{u}_i}{N+1-i'}$ in the  solution that minimizes $\sum_{i=1}^{N} \log u_i$.
\end{proof}

We will relax the denominators on the right-hand side of Equation \eqref{eqn:binary-utility} to be $N$ in the subsequent analysis.
We get that:
\begin{equation}
    \label{eqn:impartial-binary-relaxed-bound}
    \sum_{i=1}^{N} \log u_i ~ \geq ~ \sum_{i = 1}^{N} \log \left( \frac{1}{i} \sum_{j = 1}^{i} \tilde{u}_i \right) - \log \frac{N^N}{N!}
    ~.
\end{equation}

If the benchmark allocation $\vec{\tilde{x}}$ was $1$-impartial, we would have $\tilde{u}_i = \big(\prod_{j=1}^N \tilde{u}_j \big)^{\frac{1}{N}}$ for all agents $1 \le i \le N$.
In that case, the above Inequality~\eqref{eqn:impartial-binary-relaxed-bound} would already imply a constant competitive ratio because it can be written as:
\[
    \log \left(\prod_{i=1}^N u_i\right)^{\frac{1}{N}} ~ \geq ~ \log \left(\prod_{i=1}^N \tilde{u}_i\right)^{\frac{1}{N}} - ~ \log \frac{N}{(N!)^{\frac{1}{N}}}
    ~,
\]
where the second term on the right-hand side is $O(1)$ by Stirling's approximation.

For $\mu$-impartial instances with an arbitrary $\mu \ge 1$, we need one more inequality given in the next lemma.
The form of the inequality is clean so we suspect that it may have already been proved in the past.
To our best effort, however, we only found Hardy's Inequality and its several variants (c.f., \citet*{HardyLP:1952}) to have a similar spirit as they also compare the function values of a sequence of non-negative real numbers and the functions values of their prefix averages;
but those inequalities do not imply our lemma.
It would be interesting to find further applications of this inequality in other problems.


\begin{lemma}
    \label{lem:binary-upper-bound}
    Suppose that $a_1, a_2, \dots, a_N$ are real numbers between $1$ and $\mu$.
    We have:
    \[
        \frac{1}{N} \sum_{i=1}^N \log a_i ~ - ~ \frac{1}{N} \sum_{i=1}^N \log \left(
        \frac{1}{i} \sum_{j=1}^i a_j \right) ~ \le ~ \log (\log \mu + 1) + 1
        ~.
    \]
\end{lemma}

\begin{proof}
    It is without loss of generality to assume that $a_1 \le a_2 \le \dots \le a_N$ since this order minimizes the second term on the left-hand side conditioned on the set of these $N$ real numbers.
    We may further assume $a_1 = 1$ because the left-hand side remains the same if $a_1, a_2, \dots, a_N$ are multiplied by the same factor.

    Our proof strategy is to consider all possible thresholds $\theta \ge 0$ and to bound the number of indices ($1 \le i \le N$) for which $\log a_i - \log \frac{1}{i} \sum_{j=1}^i a_j$ is at least $\theta$.
    First, we rewrite the left-hand side as:
    \[
        \int_{0}^{\infty} \frac{\sum_{i = 1} ^{N} \mathbf{1} \big(\, \log a_i - \log (\frac{1}{i} \sum_{j=1}^i a_j) \geq \theta \,\big)}{N} \dif{\theta}
        ~,
    \]
    where $\mathbf{1}(\cdot)$ is the indicator function.
    
    For any $\theta \geq 0$, suppose that $i_1 < i_2 < \cdots < i_k$ are the indices for which
    \[
        \log a_i - \log \left( \frac{1}{i} \sum_{j=1}^i a_j \right) \geq \theta
        ~.
    \]

    First, by $a_1 = 1 < a_{i_1}$ and $a_j \le a_{i_1}$ for all $j \le i$, the left-hand side above is strictly smaller than $\log i_1$ when $i = i_1$.
    This gives $i_1 > e^\theta$.

    For ease of exposition, in the remaining argument we define $i_0 = 1, a_{i_0}=1$.
    Then, we have that for any $j \ge 1$:
    \begin{align*}
        \frac{i_j}{e^{\theta}}a_{i_j} \ge & ~ \left(a_1 + a_2 + \cdots + a_{i_j}\right) \\
        \geq & ~ (i_1 - i_0) a_{i_0} + (i_2 - i_1) a_{i_1} + \cdots + (i_j - i_{j-1}) a_{i_{j-1}} + a_{i_j}
        ~.
    \end{align*}

    Subject to these inequalities, $a_{i_k}$ achieves its minimum value when equalities hold everywhere.
    In that case, for any $j \ge 1$ we have (by taking the difference of two consecutive equalities):
    \[
        \frac{i_{j+1}}{e^{\theta}} a_{i_{j+1}} - \frac{i_j}{e^{\theta}} a_{i_j} = (i_{j+1} - i_j - 1) a_{i_j} + a_{i_{j+1}}
        ~.
    \]

    Therefore:
    \begin{align*}
        a_{i_{j+1}} = & ~\left( e^{\theta} - \frac{ (e^{\theta} - 1) (\frac{i_j}{e^{\theta}} - 1) }{ \frac{i_{j+1}}{e^{\theta}} - 1 } \right) a_{i_j}  \\
        \geq & ~ \left( e^{\theta} - \frac{ (e^{\theta} - 1) (\frac{i_j}{e^{\theta}} - 1) }{ \frac{i_j+1}{e^{\theta}} - 1 } \right) a_{i_j}
        &&
        \text{($i_{j+1} \geq i_j + 1$)} \\
        \geq & ~ \left( e^{\theta} - \frac{ (e^{\theta} - 1) (\frac{N}{e^{\theta}} - 1) }{ \frac{N+1}{e^{\theta}} - 1 } \right) a_{i_j} 
        &&
        \text{($i_j \leq N$)} \\
        = & ~ \frac{1}{1 - \frac{e^{\theta} - 1}{N}} a_{i_j} \\
        \geq & ~ e^{\frac{e^{\theta} - 1}{N}} a_{i_j}
        ~.
        && \text{($1-y \leq e^{-y}$)}
    \end{align*}

    For $a_{i_1}$, we have $a_{i_1} = \frac{e^\theta}{i_1} (i_1-1 + a_{i_1})$ which means (recall that $i_1 > e^\theta$):
    \[
        a_{i_1} = \frac{(i_1-1) e^\theta}{i_1 - e^\theta}  \ge \frac{N e^\theta}{(N+1) e^{-\theta} - 1} = \frac{e^\theta}{1 - \frac{e^\theta-1}{N}} \ge \frac{1}{1 - \frac{e^\theta-1}{N}} \ge e^{\frac{e^\theta-1}{N}}
        ~.
    \]

    Combining the above two inequalities, we get that $a_{i_k} \ge e^{\frac{k}{N} (e^\theta-1)}$.
    On the other hand, we have $a_{i_k} \le \mu$.
    Hence:
    \[
        \frac{k}{N} \le \frac{\log \mu}{e^\theta - 1}
        ~.
    \]
    
    Therefore, the left-hand side of the lemma's inequality is at most:
    \begin{align*}
        \int_{0}^{\infty} \min\left\{\frac{\log \mu}{e^{\theta} - 1}, 1 \right\} \dif{\theta}
        &
        = \int_{0}^{\log(\log \mu + 1)} 1 \dif{\theta} + \int_{\log(\log \mu + 1)}^{\infty} \frac{\log \mu}{e^{\theta} - 1} \dif{\theta} \\
        &
        = \log (\log \mu + 1) + \log \mu \cdot \log \left( 1 + \frac{1}{\log\mu} \right) \\[1ex]
        &
        \le \log (\log \mu + 1) + 1 
        ~,
    \end{align*}
    where the last inequality follows by $\log(1+y) \le y$ for any $y \ge 0$.
\end{proof}

Lemma~\ref{lem:binary-greedy-ratio} then follows by combining Equation~\eqref{eqn:impartial-binary-relaxed-bound}, Stirling's approximation, and Lemma~\ref{lem:binary-upper-bound} with $a_i = \frac{\tilde{u}_i}{\min_{1 \le j \le N} \tilde{u}_j}$ and $\mu = \tilde{\mu}$.

\subsection{Algorithm with a Known Upper Bound of the Impartiality Ratio}
\label{sec:greedy-reduction-bounded}

Suppose that we are given an upper bound $\mu$ of the instance's impartiality ratio $\mu^*$.
In other words, we know that the instance is $\mu$-impartial.
For any item $t$, let $\bar{v}_t=\max_{i\in[N]}v_{it}$ denote the maximum value of any agent for item $t$. 
By Lemma \ref{lem:maximum-mu-value}, we should only allocate item $t$ to the agents whose values for item $t$ is at least $\frac{1}{\mu}\bar{v}_t$.
A naïve approach is to treat all such agents as if they had value $\bar{v}_t$ for the item, and 
then apply the Myopic Greedy algorithm;
doing so would lose an extra factor $\mu$ in the competitive ratio in the worst case.
Instead, we divide the item equally into $\lceil \log \mu \rceil$ sub-items, each with supply $\frac{s_t}{\lceil \log \mu \rceil}$.
For the $j$-th sub-item, which we referred to as sub-item $(t, j)$, the agents' values are rounded down to either $\frac{\bar{v}_t}{2^j}$ or $0$.
That is, an agent $i$ has value $\frac{\bar{v}_t}{2^j}$ for sub-item $(t, j)$ if its original value for item $t$ is at least as much, and has value $0$ for this sub-item otherwise.
Then, we run the Myopic Greedy algorithm to allocate the sub-items.
See Algorithm~\ref{alg:greedy-reduction-mu} for a formal definition.

\begin{algorithm}[t] 
    \caption{\textbf{Greedy with Rounded Values}\ (for $\mu$-impartial instances)}
    \label{alg:greedy-reduction-mu}

\For {\text{\rm each item $1 \le t \le T$}}
{
    Let $\bar{v}_t = \max_{1 \le i \le N} v_{it}$
    
    \For{j = 1 to $\lceil \log \mu \rceil$}
    {
        Let there be a sub-item $(t, j)$ with supply $\frac{s_t}{\lceil \log \mu \rceil}$.

        Let each agent $i$'s value for the sub-item be:
        \[
            v_{i(t,j)}=
            \begin{cases}
                \frac{\bar{v}_t}{2^j} & \mbox{if $v_{it}\geq \frac{\bar{v}_t}{2^j}$} \\[1ex]
                0 & \mbox{otherwise.}
            \end{cases}
        \]
    }
    
    Let Algorithm~\ref{alg:greedy-binary} allocate the sub-items and let the allocation be $(x_{i(t,j)})_{1 \le i \le N, 1 \le j \le \lceil \log \mu \rceil}$.
    
    Allocate $x_{it} = \sum_{j=1}^{\lceil \log \mu \rceil} x_{i(t,j)}$ amount of item $t$ to each agent $1 \le i \le N$.
}
\end{algorithm}

\begin{theorem}
    \label{thm:greedy-reduction-mu}
    Algorithm \ref{alg:greedy-reduction-mu} is $O(\log^2 \mu)$-competitive.
\end{theorem}

\begin{proof}
    Recall that $\vec{x^*} = (x^*_{it})_{1 \le i \le N, 1 \le t \le T}$ denote the Nash welfare maximizing allocation.
    We will consider a corresponding feasible allocation $(x'_{i(t,j)})$ of the sub-items as follows:
    \[
        x'_{i(t,j)} =
        \begin{cases}
            \frac{x_{it}^*}{\lceil \log \mu \rceil} & \mbox{if $\frac{\bar{v}_t}{2^j}\leq v_{it}<\frac{\bar{v}_t}{2^{j-1}}$} \\[1ex]
            0 & \mbox{otherwise.}
        \end{cases}
    \]


    Suppose that an agent $i$'s value for item $t$ satisfies $\frac{\bar{v}_t}{2^\ell} \le v_{it} < \frac{\bar{v}_t}{2^{\ell-1}}$ for integer $\ell \ge 1$.
    Then its utility for the sub-items $(t, j)$, $1 \le j \le \lceil \log \mu \rceil$, allocated to it would be:
    \[
        \sum_{j=1}^{\lceil \log \mu \rceil} x'_{i(t,j)}v_{i(t,j)} = \sum_{j=\ell}^{\lceil \log \mu \rceil} x'_{i(t,j)}v_{i(t,j)} = \frac{x^*_{it}}{\lceil \log \mu \rceil} \sum_{j=\ell}^{\lceil \log \mu \rceil} \frac{\bar{v}_t}{2^j}
        ~.
    \]

    The summation on the right-hand-side is at least $\frac{\bar{v}_t}{2^j}$ and at most $\frac{\bar{v}_t}{2^{j-1}}$.
    Thus, we conclude that:
    \begin{equation*}
        \frac{x_{it}^*v_{it}}{2\lceil \log \mu \rceil} \leq\sum_{j=1}^{\lceil \log \mu \rceil} x'_{i(t,j)}v_{i(t,j)} \leq \frac{x_{it}^*v_{it}}{\lceil \log \mu \rceil}.
    \end{equation*}

    The agents' utilities $\vec{u'}$ for the sub-items satisfy that for any agent $1 \le i \le N$:
    \begin{equation*}
        \frac{u_i^*}{2\lceil \log \mu \rceil} \leq u'_i \leq \frac{u_i^*}{\lceil \log \mu \rceil}.
    \end{equation*}

    This implies two properties.
    First, the allocation of sub-items is $2\mu$-impartial since the original allocation $\vec{x^*}$ is $\mu$-impartial.
    Second, the Nash welfare of the allocation of sub-items is at least an $\Omega(\frac{1}{\log \mu})$ fraction of the optimal Nash welfare of the original instance.
    Therefore, by Lemma \ref{lem:binary-greedy-ratio}, we get an $O(\log \mu)$-competitive allocation comparing to the Nash welfare of sub-item allocation, which implies an $O(\log^2 \mu)$-competitive allocation w.r.t.\ the original instance.
\end{proof}

\subsection{Guessing the Impartiality Ratio}
\label{sec:greedy-reduction-general}

This is almost verbatim to the counterpart for balanced instances.
When we have no prior knowledge of the impartiality ratio, we can guess an upper bound by sampling from an appropriate distribution.
Since the final ratio depends logarithmically on the upper bound, it suffices to make a good enough guess that is at most a polynomial of the true impartiality ratio $\mu^*$.

We shall again consider a sequence of numbers starting from $2$, such that each sequel number is the square of the previous number.
We will sample each number $\mu$ with a probability that is inverse polynomial in $\log\log \mu$;
this ensures that the correct guess is made with a sufficiently large probability.
Finally, we apply the prior-dependent Greedy with Rounded Values algorithm (Algorithm~\ref{alg:greedy-reduction-mu}) with the guessed upper bound $\mu$.
See Algorithm~\ref{alg:greedy-reduction-general} for a formal definition.

\begin{algorithm}[t]
    \caption{\textbf{Greedy with Rounded Values}\ (for  instances with unknown impartiality ratio)}
    \label{alg:greedy-reduction-general}
    
    Sample $\mu$ such that it equals $2^{2^k}$ with probability $\frac{6}{\pi^2} \cdot \frac{1}{(k+1)^2}$ for any non-negative integer $k$.
    
    Run Algorithm~\ref{alg:greedy-reduction-mu} with $\mu$ to allocate the items to the agents.
\end{algorithm}

\begin{theorem}
    \label{thm:greedy-reduction-general}
    Algorithm \ref{alg:greedy-reduction-general} is $O\big( \log^2\mu^* (\log \log \mu^*)^2 \big)$-competitive.
\end{theorem}

\begin{proof}
    Suppose that $k$ is the smallest positive integer such that the balance ratio $\mu^*$ is at most $2^{2^k}$.
    Then, we have $2^{2^{k-1}} < \mu^* \le 2^{2^k}$.
    This further implies that $\log \log \lambda^* > k - 1$.
    Hence, the algorithm correctly guesses $\mu = 2^{2^k}$ with probability at least $\Omega(\frac{1}{k^2}) = \Omega( \frac{1}{(\log\log \mu^*)^2})$.
    When that happens, the algorithm's Nash welfare is an $O(\log^2 \mu^* N)$ approximation to the optimal Nash welfare.
    Therefore, even if the algorithm got zero Nash welfare from the other guesses, we would still have the stated competitive ratio.
\end{proof}

\section{Lower Bounds}

We first restate a lower bound by \citet{BanerjeeGGJ:SODA:2022} under our model.

\begin{lemma}[\citet{BanerjeeGGJ:SODA:2022}, Theorem 3]
    \label{lem:lower-bound-BanerjeeGGJ}
    No online algorithm can achieve a competitive ratio better than $\log^{1-o(1)} N$,%
    \footnote{The original theorem by \citet{BanerjeeGGJ:SODA:2022} only claimed a weaker bound of $\log^{1-\varepsilon} N$ but their proof actually showed a slightly stronger bound that we restate here.}
    even for $1$-balanced instance.
\end{lemma}


Next we prove that the logarithmic dependence on the balance ratio or the impartiality ratio is necessary.
The main ingredient is the instance given by Table~\ref{tab:lower-bound}, which is a variant of another hard instance by \citet{BanerjeeGGJ:SODA:2022} (c.f., Theorem 4 therein).
The instance has $N$ agents, and $T = N$ items of unit supplies.
Agent $i$'s value for item $t$ is $N^{2i}$ if $t \ge i$, and is $0$ otherwise.
Recall that for lower bounds it suffices to consider deterministic algorithms.
Hence, we may assume without loss of generality that agent $i$ receives the least amount of item $t = i$ among agents $i$ to $N$ in the algorithm's allocation.

\begin{table}[t]
    \caption{Illustration of a hard instance. Rows are items and columns are agents. The number in the intersection of the $t$-th row and the $i$-th column is agent $i$'s value for item $t$.}
    \label{tab:lower-bound}
    \centering
    \begin{tabular}{cccccc}
        \hline
        $t$ & Agent 1& Agent 2 & Agent 3& $\cdots$ &Agent $n$ \\
        \hline
        1 & $n^2$ & $n^2$ & $n^2$ & & $n^2$ \\
        2 & 0 & $n^4$ & $n^4$ & & $n^4$ \\
        3 & 0 & 0 & $n^6$ & & $n^6$ \\
        \multicolumn{6}{c}{$\cdots$} \\
        $n$ & 0 & 0 & 0 & & $n^{2n}$\\
        \hline
    \end{tabular}
\end{table}

\begin{lemma}
    \label{lem:lower-bound-competitive-ratio}
    No algorithm can be better than $\frac{n-1}{e}$ competitive for the instance given by Table~\ref{tab:lower-bound}.
\end{lemma}

\begin{proof}
    For any $1 \le i \le n$, by the assumption that agent $i$ receives the least amount of item $t = i$ among agents $i$ to $n$, we get that:
    \[
        x_{ii} \le \frac{1}{n-i+1}
        ~.
    \]

    Therefore, agent $i$'s utility is at most
    \[
        u_i \leq \sum_{t=1}^{i-1} n^{2t} + \frac{n^{2i}}{n-i+1} \leq \frac{n^{2i}}{n-i+1} \cdot \frac{n}{n-1}.
    \]

    On the other hand, if we allocate each item $1 \le t \le n$ to agent $i = t$ (which is the optimal allocation), agent $i$'s utility would be:
    \[
        u^*_i = n^{2i}
        ~.
    \]

    
    
    Therefore, the competitive ratio of the algorithm is at best:
    \[
        \left( \prod_{i=1}^n \frac{u^*_i}{u_i} \right)^{\frac{1}{n}} 
        \ge
        \left( \prod_{i=1}^n \frac{n^{2i}}{\frac{n^{2i}}{n-i+1} \cdot \frac{n}{n-1}} \right)^{\frac{1}{n}}
        =
        \frac{n-1}{n} \big( n! \big)^{\frac{1}{n}} > \frac{n-1}{e}
        ~.
    \]
    
\end{proof}

\begin{lemma}
    \label{lem:lower-bound-balance-impartiality-ratios}
    The imbalance ratio and impartiality ratio of the instance given by Table~\ref{tab:lower-bound} are at most $n^{2n}$.
\end{lemma}

We can now derive our lower bounds as corollaries of the above lemmas.

\begin{theorem}
    \label{thm:lower-bound-balance}
No algorithm can be better than $\min \big\{ \log^{1-o(1)} N \lambda^* ~,~ \frac{N-1}{e} \big\}$-competitive.
\end{theorem}

\begin{proof}
    If $N \ge \lambda^*$, the stated lower bound is $\log^{1-o(1)} N$ which follows by Lemma~\ref{lem:lower-bound-BanerjeeGGJ}.

    If $N < \lambda^*$ but $N^{2N} \ge \lambda^*$, the stated lower bound is $\log^{1-o(1)} \lambda^*$. 
    Consider the instance given by Table~\ref{tab:lower-bound} choosing $n = \Theta\big(\frac{\log \lambda^*}{\log\log \lambda^*}\big)$ such that $\lambda^* = n^{2n}$.
    Note that this implies $N \ge n$.
    Create $\frac{N}{n}$ copies of this instance so that the number of agents becomes $N$.
    Then, the stated lower bound follows by Lemmas~\ref{lem:lower-bound-competitive-ratio} and \ref{lem:lower-bound-balance-impartiality-ratios}.

    Finally if $N^{2N} < \lambda^*$.
    The stated bound is $\frac{N-1}{e}$, which follows by Lemmas~\ref{lem:lower-bound-competitive-ratio} and \ref{lem:lower-bound-balance-impartiality-ratios}, choosing $n = N$ in the instance given by Table~\ref{tab:lower-bound}.
\end{proof}

\begin{theorem}
    \label{thm:lower-bound-impartial}
No algorithm can be better than $\min \big\{ \log^{1-o(1)} \mu^* ~,~ \frac{N-1}{e} \big\}$-competitive.
\end{theorem}

\begin{proof}
    If $N^{2N} \ge \mu^*$, the stated lower bound is $\log^{1-o(1)} \mu^*$. 
    Consider the instance given by Table~\ref{tab:lower-bound} choosing $n = \Theta\big(\frac{\log \mu^*}{\log\log \mu^*}\big)$ such that $\mu^* = n^{2n}$.
    Note that this implies $N \ge n$.
    Create $\frac{N}{n}$ copies of this instance so that the number of agents becomes $N$.
    Then, the stated lower bound follows by Lemmas~\ref{lem:lower-bound-competitive-ratio} and \ref{lem:lower-bound-balance-impartiality-ratios}.

    Finally if $N^{2N} < \mu^*$.
    The stated lower bound is $\frac{N-1}{e}$, which follows by Lemmas~\ref{lem:lower-bound-competitive-ratio} and \ref{lem:lower-bound-balance-impartiality-ratios}, choosing $n = N$ in the instance given by Table~\ref{tab:lower-bound}.
\end{proof}

We remark that Theorem~\ref{thm:lower-bound-impartial} still holds even if the agents' values for the items are either $0$ or $1$.
For each item $t$ in Table~\ref{tab:lower-bound}, instead of letting it have unit supply and letting agents $t$ to $n$ have values $n^{2t}$ for it, we can instead let it have supply $n^{2t}$ and let agents $t$ to $n$ have values $1$ for it.
Scaling the allocation of item $t$ in the original instance by $n^{2t}$ factor would give an allocation for the $0$-$1$ value version, with the same utilities for all agents.

\section{Conclusion}

In this paper we initiate the study of online Nash welfare maximization \textit{without predictions}. We define \textit{balance ratio} and \textit{impartiality ratio} for an instance, and design online algorithms whose competitive ratios only depend on the logarithms of the aforementioned ratios of agents' utilities and the number of agents. 
One possible extension of our work is to close the gap between the lower and upper bounds in general instances.
In addition to that, our technique may provide an insight on how to remove predictions for other online problems.

\bibliographystyle{plainnat}
\bibliography{reference}

\end{document}